\documentclass[10pt]{llncs}
\def\nottoobig#1{{\hbox{$\left#1\vcenter to1.111\ht\strutbox{}\right.\n@space$}}}
\newtheorem{fact}{Fact}
\newtheorem{cclaim}{Claim}
\if01
\newtheorem{fact}{Fact}

\newtheorem{theorem}{Theorem}[section]

\newtheorem{lemma}[theorem]{Lemma}

\newtheorem{definition}[theorem]{Definition}

\fi
\usepackage{epsfig}
\usepackage{amsmath}
\usepackage{amsfonts}

\newcommand{\nat}{{\mathbb N}}

\newcommand{\poly}{{\rm poly}}

\def\nottoobig#1{{\hbox{$\left#1\vcenter
to1.111\ht\strutbox{}\right.\n@space$}}}

\newcommand{\prob}{{\rm Prob}}

\newcommand{\ie}{$\mbox{i.e.}$}

\newlength{\filength}
\newsavebox{\gcbox}
\sbox{\gcbox}{\framebox[\filength]{\rule{0ex}{2ex}}}

\newcommand{\qedblob}{\mbox{\rule[-1.5pt]{5pt}{10.5pt}}}
\def\literalqed{{\ \nolinebreak\hfill\mbox{\qedblob\quad}}}

\def\qed{\literalqed}

\newcommand{\singlespacing}{\let\CS=
\@currsize\renewcommand{\baselinestretch}{1}\tiny\CS}
\newcommand{\singlespacingplus}{\let\CS=
\@currsize\renewcommand{\baselinestretch}{1.25}\tiny\CS}
\newcommand{\doublespacing}{\let\CS=
\@currsize\renewcommand{\baselinestretch}{1.75}\tiny\CS}
\newcommand{\draftspacing}{\let\CS=
\@currsize\renewcommand{\baselinestretch}{2.0}\tiny\CS}

\def\zo{\{0,1\}}

\def\mapping{\rightarrow}

\newcommand{\zon}{\zo^n}

\newcommand{\F}{{\mathbb{F}}}

\if01
\fi

\if
\fi

\makeatletter
\def\@listI{\leftmargin\leftmargini \parsep 4.5pt plus 1pt minus 1pt\topsep6pt plus 2pt minus 2pt \itemsep  2pt plus 2pt minus 1pt}

\let\@listi\@listI
\@listi
\makeatother

\author{ {Marius Zimand\/}
\thanks{  Department of Computer and Information Sciences, Towson University,
Baltimore, MD.; email: mzimand@towson.edu; http://triton.towson.edu/\~{ }mzimand.
The author is supported in part
by NSF grant CCF 0634830.}}

\author{
{Marius Zimand}\inst{1}
\thanks{ {  The author is supported in part by NSF grant CCF 0634830. \tt  http://triton.towson.edu/\~{ }mzimand}.}}
\institute{
{Department of Computer and Information Sciences, Towson University,
Baltimore, MD, USA}
}

\date{ }

\title{On generating independent random  strings
}

\begin{document}

\maketitle

\begin{abstract}
It is shown that from two strings that are partially random and independent (in the sense of Kolmogorov complexity) 
it is possible to effectively construct polynomially many strings that are random and pairwise independent. If the two
initial strings are random, then the above task can be performed in polynomial time.  It is also possible to construct in polynomial time a random string, from two strings that have constant randomness rate.
\end{abstract}

\smallskip

\section{Introduction}
This paper belongs to a line of research that investigates whether certain attributes of randomness can be improved effectively. We focus on
finite binary strings and we regard randomness from the point of view of Kolmogorov complexity. Thus, the amount of randomness in a binary string $x$
is given by $K(x)$, the Kolmogorov complexity of $x$ and the randomness rate of $x$ is defined as $K(x)/|x|$, where $|x|$ is the length of $x$. Roughly speaking, a string $x$ is considered to be random if its randomness rate is approximately equal to $1$.
It is obvious that randomness cannot be created from nothing (e.g., from the empty string). On the other hand, it might be possible that if we
already
possess  some randomness, we can produce ``better" randomness or ``new" randomness. For the case when we start with \emph{one} string $x$, it is known that there exists no computable function that produces another string $y$ with higher randomness rate (\ie, ``better" randomness), and it is also clear that there is no computable function that produces ``new" randomness, by which we mean a string $y$ that has non-constant Kolmogorov complexity conditioned by $x$. In fact, Vereshchagin and Vyugin~\cite[Th. 4]{ver-vyu:j:kolm} construct a string $x$ with high Kolmogorov complexity so that any shorter string that has small Kolmogorov complexity conditioned by $x$ (in particular any string effectively constructed from $x$) has small Kolmogorov complexity unconditionally. 
Therefore, we need to analyze what is achievable if we start with two or more strings that have a certain amount of randomness and a certain degree of independence. In this case, in certain circumstances, positive solutions exist. For example, Fortnow, Hitchcock, Pavan, Vinodchandran and Wang~\cite{fhpvw:c:extractKol}  show that, for any $\sigma$ there exists a constant $\ell$ and a polynomial-time procedure
that from an input consisting of $\ell$ $n$-bit strings $x_1, \ldots, x_\ell$, each with Kolmogorov complexity at least $\sigma n$, constructs
an $n$-bit string with Kolmogorov complexity $\succeq n - {\rm dep}(x_1, \ldots, x_\ell)$ (${\rm dep}(x_1, \ldots, x_\ell)$ measures the dependency of the input strings and is defined as $\sum_{i=1}^{\ell} K(x_i) - K(x_1 \ldots x_\ell)$; $\succeq$ means that the inequality holds within an error of $O(\log n)$).

In this paper we focus on the case when the input consists of \emph{two} strings $x$ and $y$ of length $n$. We say that
$x$ and $y$ have dependency at most $\alpha(n)$ if the complexity of each string does not decrease by more than $\alpha(n)$ when it is conditioned by the other string, \ie, if $K(x) - K(x \mid y) \leq \alpha(n)$ and  $K(y) - K(y \mid x) \leq \alpha(n)$. The reader should have in mind the situation $\alpha(n) = O(\log n)$, in which case we say
 that $x$ and $y$ are independent (see~\cite{cal-zim:c:dlt08} for a discussion of independence for finite binary strings and infinite binary sequences). We address the following two questions:
\smallskip 
 
\emph{Question 1.} Given $x$ and $y$ with a certain amount of randomness and a certain degree of independence, is it possible to effectively/efficiently construct a string $z$ that is random?
\smallskip

	\emph{Question 2.} (a more ambitious version of Question 1)  Given $x$ and $y$ with a certain amount of randomness and a certain degree of independence, is it possible to effectively/efficiently construct strings that are random and have small dependency with $x$, with $y$, and pairwise among themselves? How many such strings exhibiting ``new" randomness can be produced?
\smallskip

A construction is \emph{effective} if it can be done by a computable function, and it is \emph{efficient} if it can be done by a polynomial-time computable function.

We first recall the well-known (and easy-to-prove) fact that if $x$ and $y$ are random and independent, then the string $z$ obtained by bit-wise XOR-ing the bits of $x$ and $y$ is random and independent with $x$ and with $y$. Our first result is an extension of the above fact.
\smallskip
 
\emph{Theorem  1.} (Informal statement.)
If $x$ and $y$ are random and have dependency at most $\alpha(n)$, then by doing simple arithmetic operations in the field GF$[2^n]$ (which take polynomial time), it is possible to produce polynomially many strings $z_1, \ldots, z_{{\rm poly}(n)}$ of length $n$ such that $K(z_i) \succeq n - \alpha(n)$ and the strings
$x$, $y$, $z_1, \ldots, z_{{\rm poly}(n)}$ are pairwise at most $\approx \alpha(n)$-dependent, 
where $\approx$ ($\succeq$)  means that the equality (resp., the inequality) is within an error of $O(\log n)$.
In particular, if $x$ and $y$ are independent, then
the output strings are random and together with the input strings form a collection of pairwise independent strings.
\smallskip

The problem is more complicated when the two input strings $x$ and $y$ have randomness rate significantly smaller than $1$. In this case, our questions are related to randomness extractors, which have been studied extensively in computational complexity.  A randomness extractor is a polynomial-time computable procedure that improves the quality of a defective source of randomness. A source of randomness is modeled by a distribution $X$ on $\zon$, for some $n$, and its defectiveness is modeled by the min-entropy of $X$ ($X$ has min-entropy $k$ if $2^{-k}$ is the largest probability that $X$ assigns to any string in $\zon$). There are several type of extractors; for us, multi-source extractors are of particular interest.   An $\ell$-multisource extractor takes as input $\ell$ defective independent distributions on the set of $n$-bit strings and outputs a string whose induced distribution is statistically close to the uniform distribution. The analogy between randomness extractors and our questions
is quite direct: The number of sources of the extractor corresponds to the number of input strings and the min-entropy of the sources corresponds to the Kolmogorov complexity of the input strings. For $\ell = 2$, the best  multisource extractors  are (a) the extractor given by Raz~\cite{raz:c:multiextract} with one source having min-entropy $((1/2) + \alpha)n$ (for some small $\alpha$) and the second source having min-entropy polylog($n$), and (b) the extractor given by Bourgain~\cite{bou:j:multiextract} with both sources having min-entropy $((1/2) - \alpha)n$ (for some small $\alpha$). Both these extractors are based on recent results in arithmetic combinatorics. It appears that finding polynomial-time constructions achieving the goals in Question 2 is difficult. If we settle for effective constructions, then positive solutions exist.
In~\cite{zim:c:kolmlimindep}, we have shown that there exists a computable function $f$ such that if $x$ and $y$ have Kolmogorov complexity $s(n)$ and dependency at most $\alpha(n)$, then $f(x,y)$ outputs a string $z$ of length
$m \approx s(n)/2$ such that $K(z \mid x) \succeq m - \alpha(n)$ and $K(z \mid y) \succeq m - \alpha(n)$.
Our second result extends the methods from~\cite{zim:c:kolmlimindep} and shows that it is possible to effectively construct polynomially many strings exhibiting ``new" randomness.
\smallskip

\emph{Theorem 2.} (Informal statement.) For every function $O(\log n) \leq s(n) \leq n$, there exists a computable function $f$ such that if $x$ and $y$ have Kolmogorov complexity $s(n)$ and dependency at most $\alpha(n)$, then $f(x,y)$ outputs polynomially many strings $z_1, \ldots, z_{{\rm poly}(n)}$ of
length $m \approx s(n)/3$ such that $K(z_i) \succeq m - \alpha(n)$ and the strings $(x,y, z_1, \ldots, z_{{\rm poly}(n)})$ are pairwise at most $\approx \alpha(n)$-dependent. In particular, if $x$ and $y$ are independent, then
the output strings are random and together with the input strings form a collection of pairwise independent strings.
\smallskip

For Question 1, we give a polynomial-time construction in case $x$ and $y$ have linear Kolmogorov complexity, \ie,
$K(x) \geq \delta n$ and $K(y) \geq \delta n$, for a positive constant $\delta > 0$. The proof relies heavily on 
a recent result of Rao~\cite{rao:c:randpm}, which shows the existence of $2$-source condensers. (A $2$-source condenser is similar but weaker than a $2$-source extractor in that the condenser's output is only required to be statistically close to a distribution that has larger min-entropy rate than that of its inputs, while the
extractor's output is required to be statistically close to the uniform distribution.)
\smallskip

\emph{Theorem 3.} (Informal statement.) For every constant $\delta > 0$, there exists a polynomial-time computable function $f$ such that
if $x$ and $y$ have Kolmogorov complexity $\delta n$ and dependency at most $\alpha(n)$, then $f(x,y)$ outputs a string $z$ of length $m = \Omega(\delta n)$ and $K(z) \geq m - (\alpha(n) + \poly (\log n))$.
\smallskip

The main proof technique is an extension of the method used in~\cite{zim:c:csr} and in~\cite{zim:c:kolmlimindep}. It uses ideas from 
Fortnow et al.~\cite{fhpvw:c:extractKol}, who showed that a multi-source extractor can also be used to extract Kolmogorov complexity.
A key element is the use of \emph{balanced tables}, which are combinatorial objects similar to $2$-source extractors. A balanced table is
an $N$-by-$N$ table whose cells are colored with $M$ colors in such a way that each sufficiently large rectangle inside the table is
colored in a balanced way, in the sense that all colors appear approximately the same number of times. The exact requirements for the balancing
property are tailored according to their application. The type of balanced table required in Theorem 2 is shown to exist using the probabilistic
method and then constructed using exhaustive search. This is why the transformation in Theorem 2 is only effective, and not polynomial-time
computable. The existence of the type of balanced table used in Theorem 3 is a direct consequence of  Rao's $2$-source condenser.

The paper is structured as follows. Sections~\ref{s:prelim} and~\ref{s:indep} introduce the notation and the main concepts of Kolmogorov complexity. 
Section~\ref{s:balancedtable} is dedicated to balanced tables. Theorem 1 and Theorem 2 are proved in Section~\ref{s:multiplerandomindep}, and Theorem 3 is
proved in Section~\ref{s:polyone}.

\subsection{Preliminaries}
\label{s:prelim}
$\nat$ denotes the set of natural numbers. For $n \in \nat$, $[n]$ denotes the set $\{1,2, \ldots, n\}$.
We work over the binary alphabet $\zo$. A string is an element of $\{0,1\}^*$.  If $x$ is a string, $|x|$ denotes its length.  The cardinality of a finite set $A$ is denoted $|A|$.  Let $M$ be a standard Turing machine. For any string $x$, define the \emph{Kolmogorov complexity} of $x$ with respect to $M$, as 
$K_M(x) = \min \{ |p| \mid M(p) = x \}$.
 There is a universal Turing machine $U$ such that for every machine $M$ there is a constant $c$ such that for all $x$,
$K_U(x) \leq K_M(x) + c$.
We fix such a universal machine $U$ and dropping the subscript, we let $K(x)$ denote the Kolmogorov complexity of $x$ with respect to $U$. For the concept of conditional Komogorov complexity, the underlying machine is a Turing machine that in addition to the read/work tape which in the initial state contains the input $p$, has a second tape containing initially a string $y$, which is called the conditioning information. Given such a machine $M$, we define the Kolmogorov complexity of $x$ conditioned by $y$ with respect to $M$ as 
$K_M(x \mid y) = \min \{ |p| \mid M(p, y) = x \}$.
Similarly to the above, there exist  universal machines of this type and a constant $c$ and they satisfy the relation similar to the one above, but for conditional complexity. We fix such a universal machine $U$, and dropping the subscript $U$, we let $K(x \mid y)$ denote the Kolmogorov complexity of $x$ conditioned by $y$ with respect to $U$. 
In this paper, the constants implied in the $O(\cdot)$ notation depend only on  the universal machine.

The Symmetry of Information Theorem (see \cite{zvo-lev:j:kol}) states that for all strings $x$ and $y$:
\begin{equation}
\label{e:symmetryinf}
| (K(x) - K(x\mid y)) - (K(y) - K(y \mid x)) | \leq O( \log K(x) + \log K(y)).
\end{equation}
In case the strings $x$ and $y$ have length $n$, it can be shown that
\begin{equation}
\label{e:symmetryinf2}
| (K(x) - K(x\mid y)) - (K(y) - K(y \mid x)) | \leq 2 \log n + O(1).
\end{equation}

Sometimes we need to concatenate two strings $a$ and $b$ in a self-delimiting matter, \ie, in a way that allows to retrieve
each one of them. A simple way to do this is by taking $a_1 a_1 a_2 a_2 \ldots a_n a_n 01 b$, where $a = a_1 \ldots a_n$, with each $a_i \in \zo$. A more efficient encoding is as follows. Let $|a|$ in binary notation be $c_1 c_2 \ldots c_k$. Note that $k = \lfloor \log |a| \rfloor + 1$. Then
we define ${\rm concat}(a,b) = c_1 c_1 c_2 c_2 \ldots c_k c_k 01 a b$. Note that $|{\rm concat}(a,b)| = |a| + |b| + 2 \lfloor \log |a| \rfloor  + 4$.

\subsection{Independent strings}
\label{s:indep}
\begin{definition}
(a) Two strings $x$ and $y$ are at most $\alpha(n)$-dependent if
$K(x) - K(x|y) \leq \alpha(|x|)$ and
$K(y) - K(y|x) \leq \alpha(|y|)$.

(b) The strings $(x_1, x_2, \ldots )$ are pairwise at most $\alpha(n)$-dependent, if for every $i \not= j$, $x_i$ and $x_j$ are at most
$\alpha(n)$-dependent. 
\end{definition}

\subsection{Balanced tables}
\label{s:balancedtable}

A table is a function $T : [N] \times [N] \mapping [M]$. In our applications, $N$ and $M$ are powers of $2$, \ie, $N = 2^n$ and $M = 2^m$. We identify $[N]$ with $\zo^n$ and $[M]$ with $\zo^m$. Henceforth, we assume this setting. 

It is convenient  to view such a function as a two dimensional table with $N$ rows and $N$ columns where each entry has a color from the set $[M]$. If $B_1, B_2$ are subsets of $[N]$, the $B_1 \times B_2$ rectangle of table $T$ is the part of $T$ comprised of the rows in $B_1$ and the columns in $B_2$. If $A \subseteq \zo^m$ and $(x,y) \in [N] \times [N]$, we say that the cell $(x,y)$ is $A$-colored if $T(x,y) \in A$.

In our proofs, we need the various tables to be \emph{balanced}, which, roughly speaking, requires that in each sufficiently
large rectangle $B_1 \times B_2$, all colors appear approximately the same number of times.

One variant of this concept is given in the following definition.
\begin{definition}
\label{d:stronglybalanced}
Let $k \in \nat$. The table $T$ is
$(S,n^k)$-strongly balanced if for every pair of sets $B_1$ and $B_2$, where  $B_1 \subseteq [N]$, $|B_1| \geq S$, $B_2 \subseteq [N]$, $|B_2| \geq S$, the following two inequalities hold:
\begin{itemize}
	\item[(1)] For every $a \in [M]$,
	\[
	|\{(x,y) \in B_1 \times B_2 \mid T(x,y) = a\}| \leq \frac{2}{M}|B_1 \times B_2|,
	\]
	\item[(2)] for every $(a,b) \in [M]^2$ and for every $(i,j) \in [n^k]^2$,
	\[
	|\{(x,y) \in B_1 \times B_2 \mid T(x+i,y) = a \mbox{ and } T(x+j,y) = b \}| \leq \frac{2}{M^2}|B_1 \times B_2|,
	\]
	where addition is done modulo $N$.
\end{itemize}
\end{definition}
Using the probabilistic method, we show that, under some settings for the parameters, strongly-balanced tables exist.
\begin{lemma}
\label{l:balancedtable}
If $S^2 > 3 M^2 \ln M + 6M^2 \cdot k \cdot \ln n + 6 SM^2 + 6 SM^2 + 6 SM^2 \ln(N/S) + 3M^2$, then there exists
an $(S,n^k)$ - strongly balanced table.

\end{lemma}

NOTE: The condition is satisfied if $M = o((1/\sqrt{n}) S^{1/2})$.
\medskip

\begin{proof}
We first fix $(a,b) \in [M]^2$, two sets $B_1$ and $B_2$ with $B_1 \subseteq [N]$, $|B_1| = S$, $B_2 \subseteq [N]$, $|B_2| = S$.
Note that for a fixed cell $(x, y) \in B_1 \times B_2$ and
fixed $j \in [n^k]$, $\prob[T(x,y) = a] = 1/M$ and  $\prob[T(x,y) = a \mbox{ and } T(x+j,y) = b] = 1/M^2$.

Therefore, by the Chernoff bounds,
\[
\prob\bigg[ \frac{\mbox{number of $a$-colored cells in $B_1 \times B_2$}}{S^2} > 2 \frac{1}{M}\bigg] \leq e^{-(1/3) (1/M) S^2},
\]
and, for fixed $j$,
\[
\prob \bigg[ \frac{\mbox{number of $(a,b)$-colored $j$-apart cells in $B_1 \times B_2$}}{S^2} > 2 \frac{1}{M^2}\bigg] \leq e^{-(1/3) (1/M^2) S^2}.
\]
There are $M$ possibilities for choosing $a$, and the number of possibilities for choosing the sets $B_1$ and $B_2$ is
${N \choose S}^2 \leq (eN/S)^{2S} = e^{2S + 2S \ln (N/S)}$. Therefore, the probability that the relation~(1) in
Definition~\ref{d:stronglybalanced} does not hold is bounded
by
\begin{equation}
\label{e:eq1}
e^{-(1/3) (1/M) S^2 + \ln M + 2S + 2S \ln (N/S)}.
\end{equation}
There are $M^2$ possibilities for choosing $(a,b)$, $n^{2k}$ possibilities for $(i,j)$ and the number of possibilities for choosing the sets $B_1$ and $B_2$ is
${N \choose S}^2 \leq (eN/S)^{2S} = e^{2S + 2S \ln (N/S)}$. Therefore, the probability that the relation~(2) in Definition~\ref{d:stronglybalanced} does not hold is bounded by
\begin{equation}
\label{e:eq2}
e^{-(1/3) (1/M^2) S^2 + 2\ln M + 2k\ln n + 2S + 2S \ln (N/S)}.
\end{equation}
If the parameters satisfy the requirement stated in the hypothesis, then the bound in Equation~(\ref{e:eq1}) is less than $e^{-1} < 1/2$ and
the bound in Equation~(\ref{e:eq2}) is less than $e^{-1}< 1/2$. Therefore the probability that both relation~(1) and relation~(2) hold is positive, and thus there exists a $(S,n^k)$- strongly balanced table.~\qed
\end{proof}

The above proof uses the probabilistic method which does not indicate an efficient way to construct such tables. In our application, we will build such tables by exhaustive search, an operation that can be done in EXPSPACE.

A weaker type of a balanced table can be constructed in polynomial-time using a recent result of  Rao~\cite{rao:c:randpm}. We first recall the
following definitions. Let $X$ and $Y$ be two probability distributions on $\zo^n$. The distributions $X$ and $Y$ are $\epsilon$-close if
for
every $A \subseteq \zo^n$, $| \prob( X \in A) - \prob(Y \in A) | < \epsilon$. The min-entropy of distribution $X$ is
$\max_{a \in \zo^n} ( \log ( 1/\prob(X=a)))$.

\begin{fact}
\label{t:condenser}
\cite{rao:c:randpm}
For every $\delta > 0$, $\epsilon > 0$, there exists a constant $c$ and a polynomial-time computable function
$Ext :\zo^n \times \zo^n \mapping \zo^m$, where $m = \Omega(\delta n)$, such that if $X$ and $Y$ are two independent
random variables taking values in $\zo^n$ and following
distributions over $\zo^n$ with min-entropy at least $\delta n$, then
$Ext(X,Y)$ is $\epsilon$-close to a distribution with min-entropy $m - (\delta \log 1/\epsilon)^c$.
\end{fact}

Rao's result easily implies the existence of a polynomial-time table with a useful balancing property.
\begin{lemma}
\label{l:balanceRao}
Let $\delta > 0$, $\epsilon > 0$ and let $c$ be the constant and $Ext : \zo^n \times \zo^n \mapping \zo^m$ be the function from 
Theorem~\ref{t:condenser}, corresponding to these parameters. We identify $\zo^n$ with $[N]$ and $\zo^m$ with $[M]$ and view $Ext$ as an $[N] \times [N]$ table colored with $M$ colors. Then for every rectangle $B_1 \times B_2 \subseteq [N] \times [N]$, where
$|B_1 | \geq 2^{\delta n}$ and $|B_2 | \geq 2^{\delta n}$ and for every $A \subseteq [M]$, the number of cells
in $B_1 \times B_2$ that are $A$-colored is at most
\[
\bigg(\frac{|A|}{M}2^{(\delta \log (1/\epsilon))^c} + \epsilon \bigg) \cdot |B_1 \times B_2|.
\]
\end{lemma}
\begin{proof}
Let $B_1$ and $B_2$ be two subsets of $\zo^n$ of size $\geq 2^{\delta n}$. Let $X$ and $Y$ be two independent random variables 
that follow the uniform distributions
on $B_1$, respectively $B_2$ and assume the value $0$ on $\zo^n - B_1$, respectively $0$ on $\zo^n - B_2$. Since $X$ and $Y$ have min-entropy
$\geq 2^{\delta n}$, it follows that $Ext(X,Y)$ is $\epsilon$-close to a distribution $Z$ on $\zo^m$ that has min-entropy
$m - (\delta \log 1/\epsilon)^c$. If $A \subseteq \zo^m$, then $Z$ assigns to $A$ probability mass at most
$\frac{|A|}{M} 2^{(\delta \log 1/\epsilon)^c}$, because it assigns to each element in $\zo^m$ at most
$2^{-(m - (\delta \log 1/\epsilon)^c)}$.
Thus, $Ext(X,Y)$ assigns to $A$ probability mass at most
$\frac{|A|}{M} 2^{(\delta \log 1/\epsilon)^c} + \epsilon$.
This means that the number of occurrences of $A$-colored cells in the $B_1 \times B_2$ rectangle is bounded by
$\bigg(\frac{|A|}{M} 2^{(\delta \log 1/\epsilon)^c} + \epsilon\bigg) \cdot |B_1 \times B_2|$.~\qed
\end{proof}

\section{Generating multiple random independent strings}
\label{s:multiplerandomindep}
We prove Theorem 1. The formal statement is as follows.

\begin{theorem}
For every $k \in \nat$, there is a polynomial-time computable function $f$ that on input 
$x_1, x_2$, two strings of length $n$,  outputs $n^k$ strings $x_3, x_4, \ldots, x_{n^k + 2}$, strings of length $n$,   with the following property. For every sufficiently large $n$ and for every function $\alpha(n)$, if $x_1$ and $x_2$ satisfy 

(i) $K(x_1) \geq n - \log n$, 

(ii) $K(x_2) \geq n - \log n$, and 

(iii) $x_1$ and $x_2$ are at most $\alpha(n)$-dependent, 

then 

(a) $K(x_i) \geq n - (\alpha(n) + (k + O(1))\log n$, for every $i \in \{3, \ldots, n^k +2\}$,
and

(b) the strings $x_1, x_2, \ldots, x_{n^k + 2}$ are pairwise at most $\alpha(n) + (3k+O(1))\log n$-dependent.
\end{theorem}
\begin{proof} Let $x_1, x_2 \in \zo^n$ be such that $K(x_1) \geq n-\log n$, $K(x_2) \geq n - \log n$.  Since
$x_1$ and $x_2$ are at most $\alpha(n)$-dependent,  $K(x_1 \mid x_2) \geq K(x_1) - \alpha(n) 
\geq n - (\alpha(n) + \log n)$. Similarly, $K(x_2 \mid x_1) \geq n - (\alpha(n) + \log n)$.

The function $f$ outputs
\[
\begin{array}{rl}
x_3 & = x_1 + 1 \cdot x_2, \\

x_4 & = x_1 + 2 \cdot x_2, \\

\vdots \\
 x_{n^k+2}& = x_1 + n^k \cdot x_2,
 \end{array}
\]
where the arithmetic is done in the finite field GF$[2^n]$ and $1, 2, \ldots, n^k$ denote the first (in some canonical ordering) $n^k$ non-zero elements of GF$[2^n]$.

Let $x_i$ be one of the ``new" strings, \ie, $i \in \{3, \ldots, n^k+2\}$. Let $t$ be defined by $K(x_i \mid x_1) = t$. Observe that given $x_1$, $i$ (that can be described with $k \log n$ bits) and $t + O(1)$ bits we can construct $x_2$; first we compute $x_i$ and then from $x_1$ and $x_i$, we derive $x_2$.

Therefore, $K(x_2 \mid x_1) \leq t + k\log n + 2(\log k + \log \log n) + O(1)$. Since $K(x_2 \mid x_1) \geq n - (\alpha(n) + \log n)$, it follows that $t \geq n - (\alpha(n) + (k+O(1))\log n)$ (taking into account that $k < n$; if $k \geq n$, the theorem holds trivially).  Therefore, $K(x_i \mid x_1) \geq n - (\alpha(n) + (k+O(1)))\log n$ (which implies (a)). 
We infer that 
\[
\begin{array}{ll}
K(x_i) - K(x_i \mid x_1) & \leq (n+O(1)) - (n - (\alpha(n) + (k+O(1))\log n)) \\
& = \alpha(n) + (k+O(1))\log n.
\end{array}
\]
By the Symmetry of Information Theorem, $K(x_1) - K(x_1 \mid x_i) \leq \alpha(n) + (k+O(1))\log n$, and thus $x_i$ and $x_1$ are at most
$\alpha(n) + (k+O(1))\log n$-dependent.

Similarly, $x_i$ and $x_2$ are at most
$\alpha(n) + (k+O(1))\log n$-dependent. Thus, (b) follows for pairs $(x_i, x_1)$ and $(x_i, x_2)$ with $i \geq 3$.

Let us next consider a pair of strings $(x_i, x_j)$ with $i \not = j$ and $i, j \in \{3, \ldots, n^k+2\}$. Let $t$ be defined by $K(x_i \mid x_j) = t$.
Note that given $x_j$, $i$ and $j$ and $t+ O(1)$ bits we can construct $x_1$: first we compute $x_i$ and then from $x_i$ and $x_j$, we deive $x_1$.  Therefore, 
\[
K(x_1 \mid x_j) \leq t + 2k\log n + 2(\log k + \log \log n) + O(1).
\]
Recall that 
\[
K(x_1) - K(x_1 \mid x_j) \leq \alpha(n) + (k+O(1))\log n.
\]
Then, 
\[
\begin{array}{ll}
t + 2k \log n + 2(\log k + \log \log n) + O(1) & \geq K(x_1 \mid x_j) \\
& \geq K(x_1) - (\alpha(n) + (k+O(1))\log n)\\
& \geq n - (\alpha(n) + (k+O(1))\log n).
\end{array}
\]
Thus, $K(x_j \mid x_i) = t \geq n - (\alpha(n) + (3k+O(1))\log n)$.
It follows that
\[
\begin{array}{ll}
K(x_j) - K(x_j \mid x_i) & \leq (n+O(1)) - (n - (\alpha(n) + (3k+O(1))\log n)) \\
& \leq \alpha(n) + (3k+O(1))\log n.
\end{array}
\]
Thus, $x_j$ and $x_i$ are at most $\alpha(n) + (3k+O(1))\log n$-dependent.~\qed

\end{proof}

We next prove Theorem 2. The formal statement is as follows.

\begin{theorem} 
For every $k \in \nat$, for every computable function $s(n)$ verifying $(6k + 15) \log n < s(n) \leq n$ for every $n$, 
there exists a computable function $f$ that, for every $n$, on input two strings $x_1$ and $x_2$ of length $n$, outputs $n^k$ strings
$x_3, x_4, \ldots, x_{n^k + 2}$ of length $m= s(n)/3 - (2k+5) \log n$ with the following property. For every sufficiently large $n$ and for every function $\alpha(n)$, if 

(i) $K(x_1) \geq s(n)$, 

(ii) $K(x_2) \geq s(n)$ and 

(iii) $x_1$ and $x_2$ are at most $\alpha(n)$ - dependent, 

then 

(a) $K(x_i ) \geq m - (\alpha(n) + O(\log n))$, for every $i \in \{3, \ldots, n^k+2\}$ and

(b) the strings in the set $\{x_1, x_2, \ldots, x_{n^k + 2} \}$ are pairwise at most $\alpha(n) + (2k+O(1))\log n$-dependent.
\end{theorem}
\begin{proof} 
We fix $n$ and let $N = 2^n$, $m= s(n)/3 - (2k+5) \log n$, $M = 2^m$, $S =2^{2s(n)/3}$. We also take $t= \alpha(n) + 7 \log n$.
The requirements of Lemma~\ref{l:balancedtable} are satisfied and therefore there exists a table $T: [N] \times [N] \mapping [M]$ that is $(S,n^k)$- strongly balanced. By brute force, we find the smallest (in some canonical sense) such table $T$. Note that the table $T$ can be described with $\log n + O(1)$ bits. 

The function $f$ outputs 
\[
\begin{array}{rl}
x_{3} & = T(x_1+1, x_2),\\
 x_{4} &= T(x_1+2, x_2), \\
 \vdots \\
  x_{n^k + 2} & = T(x_1+n^k, x_2).

\end{array}
\]
We show the following two claims.

\begin{cclaim}
\label{c:claim1}
For every $j \in \{3, \ldots, n^k+2\}$,  $K(x_j \mid x_1) \geq K(x_j) - (\alpha(n) + O(\log n))$ and $K(x_j \mid x_2) \geq K(x_j) - (\alpha(n) + O(\log n))$.
\end{cclaim}
\begin{cclaim}
\label{c:claim2}
For every $i, j \in \{3, \ldots, n^k+2\}$,  $K(x_j \mid x_i) \geq K(x_j) - (\alpha(n) + (2k+O(1))) \log n$. 
\end{cclaim}

Claim~\ref{c:claim1} is using ideas from the paper~\cite{zim:c:kolmlimindep}.  For the sake of making this paper self-contained we present the proof.
Let $j \in \{3, \ldots, n^k + 2\}$. We show that  $K(x_j \mid x_1)$ and $K(x_j \mid x_2)$ are at least $m - \alpha(n) - 7 \log n$. We show this relation for $K(x_j \mid x_2)$ (the proof for $K(x_j \mid x_1)$ is similar). Suppose that $K(x_j \mid x_2) <  m - \alpha(n) - 7 \log n = m - t$. Let $t_1 = K(x_1)$. Note that $t_1 \geq s(n)$. Let $B = \{u \in \zo^n \mid K(u) \leq t_1\}$. Note that $2^{t_1 + 1} > |B| \geq 2^{2s(n)/3} = S$. ($B$ has size $\geq 2^{2s(n)/3}$ because it contains the set $0^{s(n)/3}\zo^{2s(n)/3}$.)We say that a column $u \in [N]$ is \emph{bad for color $a \in [M]$ and $B$} if
the number of occurrences of $a$ in the $B \times \{u\}$ subrectangle of $T$ is greater that $(2/M) \cdot |B|$ and we say that $u$ is \emph{bad for $B$} if it is bad for some color $a$ and $B$. For every $a \in [M]$, the number of $u$'s that are bad for $a$ and $B$ is $< S$ (since $T$ is $(S,n^k)$-strongly balanced and we can take into account the first balancing property of such tables). Therefore, the number of $u$'s that are bad for $B$ is $< M \cdot S$. Given $t_1$ and a description of the table $T$, one can enumerate the set of $u$'s that are bad for $B$. This implies that any $u$ that is bad for $B$ can be described by its rank in this enumeration and the information needed to perform the enumeration. Therefore, if $u$ is bad for $B$,
\[
\begin{array}{ll}
K(u) & \leq \log (M \cdot S) + 2 (\log t_1 + \log n)+O(1) \\
& \leq m + 2s(n)/3 + 4\log n +O(1) \\
& < s(n),
\end{array}
\]
provided $n$ is large enough.
Since $K(x_2) \geq s(n)$, it follows that $x_2$ is good for $B$.

Let $A = \{w \in [M] \mid K(w \mid x_2) < m- t\}$.  We have $|A| < 2^{m-t}$ and, by our assumption, $x_j \in A$. Let $G$ be the subset of $B$ of positions in the strip $B \times \{x_2\}$ of $T$ having a color from $A$ (formally, $G= {\rm proj}_1 (T^{-1}(A) \cap (B \times \{x_2\})$) . Note that $x_1$ is in $G$. Each color $a$ occurs in the strip $B \times \{x_2\}$ at most $(2/M) \cdot |B|$ (because $x_2$ is good for $B$). Therefore the size of $G$ is bounded
by 
\[
|A| \cdot(2/M) \cdot |B| < 2^{m-t} \cdot(2/M) \cdot 2^{t_1+1} \leq 2^{t_1 - t + 2}.
\]
Given $x_2, t_1, m-t$ and a description of the table $T$, one can enumerate the set $G$. Therefore, $x_1$ can be described by its rank in this enumeration and by the information needed to perform the enumeration. It follows that
\[
\begin{array}{ll}
K(x_1 \mid x_2) & \leq t_1 - t + 2 + 2(\log t_1 + \log(m-t) + \log n) + O(1) \\
& \leq t_1 - t + 6 \log n +O(1) \\
& = t_1 - \alpha(n) - \log n + O(1)\\
& = K(x_1) - \alpha(n) - \log n +O(1),
\end{array}
\]
which contradicts that $x_1$ and $x_2$ have dependency at most $\alpha(n)$.
\smallskip

We next prove Claim~\ref{c:claim2}.
\smallskip

We fix two elements $i \not= j$ in $\{3, \ldots, k+2\}$ and analyze $K(x_i | x_j)$. 

Let $t_1 = K(x_1)$ and $t_2 = K(x_2)$. From hypothesis, $t_1 \geq s(n)$ and $t_2 \geq s(n)$. 
We define $B_1 = \{u \in \zon \mid K(u) \leq t_1\}$ and $B_2 = \{u \in \zon \mid K(u) \leq t_2\}$. We have $S \leq |B_1| < 2^{t_1+1}$ and $S \leq |B_2| < 2^{t_2+1}$. ($B_1$ and $B_2$ have size larger than $S = 2^{2s(n)/3}$, because they contain the set
$0^{s(n)/3} \{0,1\}^{2s(n)/3}$.)

Let $T_{i,j}^{-1}(x_i, x_j)$ denote the set of pairs $(u,v) \in [N] \times [N]$ such that $T(u+i,v) = x_i$ and $T(u+j,v) = x_j$.

Note that $(x_1, x_2) \in T_{i,j}^{-1} (x_i, x_j) \cap (B_1 \times B_2)$. Since the table $T$ is strongly balanced,
\[
|T_{i,j}^{-1} (x_i, x_j) \cap (B_1 \times B_2)| \leq \frac{2}{2^{-2m}} 2^{t_1 + t_2+2} = 2^{t_1+t_2-2m+3}.
\]
Note that $T_{i,j}^{-1} (x_i, x_j) \cap (B_1 \times B_2)$ can be effectively enumerated given
$x_i$, $x_j$, $i$, $j$, and the table $T$. Thus $x_1x_2$ can be described from $x_i x_j$, the rank of $(x_1, x_2)$ in the above enumeration, $i$, $j$, and the table $T$.
This implies that 
\[
\begin{array}{ll}
K(x_1x_2) & \leq t_1 + t_2 - 2m + 3 +K(x_ix_j)+ 2k\log n + 2(\log k + \log \log n) + O(\log n) \\
& \leq t_1 + t_2 - 2m + K(x_ix_j) + (2k +O(1)) \log n.
\end{array}
\]
On the other hand, $K(x_1 x_2) \geq K(x_1) + K(x_2 \mid x_1) - O(\log n)$ and $K(x_2 \mid x_1) \geq K(x_2) - \alpha(n)$.
Therefore,
\[
\begin{array}{ll}
K(x_1x_2) & \geq K(x_1) + K(x_2) - (\alpha(n) + O(\log n)) \\    
& = t_1 + t_2 - (\alpha(n) + O(\log n)).

\end{array}
\]
Combining the last two inequalities, we get that
\[
t_1 + t_2 -(\alpha(n) + O(\log n)) \leq t_1 + t_2 - 2m + K(x_i x_j) + (2k +O(1)) \log n
\]
which implies that 
\[
K(x_i x_j) \geq 2m - \alpha(n) - (2k +O(1)) \log n.
\]
Therefore
\[
\begin{array}{ll}
K(x_j | x_i) & \geq K(x_i x_j) - K(x_i) - O(\log n) \\
& \geq (2m - \alpha(n) - (2k +O(1)) \log n) - (m+ O(1)) - O(\log n) \\
& = m - \alpha(n) - (2k +O(1)) \log n.
\end{array}
\]
It follows that
\[
\begin{array}{ll}
K(x_j) - K(x_j \mid x_i) & \leq (m+O(1)) - (m - \alpha(n) - (2k +O(1)) \log n) - O(\log n) \\
& \leq \alpha(n) + (2k+ O(1)) \log n.
\end{array}
\]
Thus, $x_j$ and $x_i$ are at most $\alpha(n) + (2k+ O(1)) \log n)$-dependent.~\qed

\end{proof}

\if01
NOTE: The stuff below is not needed.
\\

Let $\F$ be the field with $2^n$ elements. Let ${\rm ind}(n)$ be the largest $k$ such that there exist $(k+2)$ vectors in
$\F^k$ with the property that any $k$ of them are linearly independent (when $\F^k$ is viewed as a vector space over $\F$).

\begin{lemma}
Let $k$ be the largest integer such that $k^3 - k < 2^{n+1}$. Then ${\rm ind}(n) \geq k$.

\end{lemma}
\begin{proof}
Fix $k$ such that $k^3 - k < 2^{n+1}$. We consider the following $k \times (k+2)$ matrix
	\[
	A(x) = \left(
		\begin{array}{cccccc} 
		1 & 1 & 1 & 0 & \ldots  & 0 \\
		x & 1 & 0 & 1 & \ldots & 0 \\
		x^2 & 1 & 0 & 0 & \ldots & 0 \\
		\vdots & & & & & \\
		x^{k-1} & 1 & 0 & 0  & \ldots & 1 \\
		\end{array} 
		\right),
		\]
\ie, the first column is $(1, x, x^2, \ldots , x^{k-1})$, the second column is $(1,1,\ldots, 1)$ and the rest of the matrix is the $k \times k$ unity matrix $I_k$. Here $x$ is a variable ranging over $\F$. We show that we can assign to $x$ a value $\alpha \in \F$ such that $A(\alpha)$ has rank $k$.

It is easy to see that any $k \times k$ submatrix of $A$ that does not include the first column has a non-zero determinant. Let us look at the
$k \times k$ submatrices that include the first column. There are ${k+1 \choose 2}$ such submatrices and the determinant of each one of them is a
polynomial of degree $(k-1)$. These polynomials have, taken together, at most $(k-1) \cdot  {k+1 \choose 2} = \frac{k^3 - k}{2}$ zeros in $\F$. Therefore, given that $\frac{k^3 - k}{2} < 2^n = |\F|$, there exists $\alpha \in \F$ that is not a root of any of these polynomials. 

It follows that $A(\alpha)$ has rank $k$.~\qed
\end{proof}
\medskip

{\bf Generating independent strings}
\medskip
\fi
\section{Polynomial-time generation of one random string} 
\label{s:polyone}
In this section we prove Theorem~3. The formal statement is as follows.
\if01
\begin{theorem}
************

WRONG, because $A_1 \subseteq A_2$ and not vice-versa and the proof breaks down.

**************
\\

For every $\delta > 0$, there exists a constant $c$ and a polynomial-time computable function
$f: \zo^n \times \zo^n \mapping \zo^m$, where $m = O(\delta n)$, with the following property. If $x$ and $y$ are two strings of length $n$ satisfying

(a) $K(x) \geq \delta n$

(b) $K(y) \geq \delta n$

(c) $x$ and $y$ are at most $\log n$ - dependent,

then
\[
\begin{array}{ll}
K(f(x,y) \mid x) & \geq m - O((\log n)^c) \mbox{ and } \\
K(f(x,y) \mid y) & \geq m - O((\log n)^c)

\end{array}
\]
\end{theorem}
\begin{proof}
Let $\epsilon = 4n^{12}$ and let $c$ be the constant and  $Ext : \zo^n \times \zo^n \mapping \zo^m$ be the function given by Lemma~\ref{l:balanceRao} for parameters $(\delta/2)$ and $\epsilon$. Let $t = 12 \log n + ((\delta/2) \log 1/\epsilon)^c + 2 = O((\log n)^c)$. 

The function $f$ on input $x$ and $y$ returns $z = Ext(x,y)$. 
We show that $K(z \mid x) \geq m-t$ and
$K(z \mid y) \geq m-t$.  Actually we prove the second inequality (the first one can be proven similarly).

Suppose $K(z \mid y) < m-t$.

Let $t_1 = K(x)$, $t_2 = K(y)$, $B_1 = \{u \in \zo^n \mid K(u) \leq t_1\}$, $B_2 = \{u \in \zo^n \mid K(u) \leq t_2\}$.
From hypothesis, $t_1 \geq \delta n$ and $t_2 \geq \delta n$.

Note also that $2^{\delta n/2} \leq |B_1| \leq 2^{t_1 +1}$ and $2^{\delta n/2} \leq |B_2| \leq 2^{t_2 +1}$.

Let
\[
A_1 = \{v \in \zo^m \mid K(v) < m-t\}
\]
and
\[
A_2 = \{v \in \zo^m \mid K(v \mid y) <  m-t\}
\]
We focus on the table defined by the function $Ext: [N] \times [N] \mapping [M]$, where, as usual, we have identified $\zo^n$ with $[N]$ and
$\zo^m$ with $[M]$.

We say that column $u \in \zo^n$ is {\em bad for $A$ and $B_1$} if the number of $A$-colored cells in the rectangle $B_1 \times \{y\}$ is
\[
\geq n^6 \cdot \bigg( \frac{|A|}{M} \cdot 2^{((\delta/2) \log (1/\epsilon))^c} + \epsilon \bigg) \cdot |B_1|.
\]
Taking into account the balancing property from Lemma~\ref{l:balanceRao}, it follows that the number of columns  that are bad for $A_1$ and $B_1$ is at most
\[
\frac{|B_2|}{n^6} \leq 2^{t_2 - 6\log n + 1}.
\]
The set of bad columns for $A_1$ and $B_1$ can be effectively enumerated given $t_1$, $\delta$ and the table (which can be described with $\log n + O(1)$ bits). Therefore any column $u$ that is
bad for $A_1$ and $B_1$ can be described by its rank in the above enumeration and by the information needed to perform the enumeration.
Consequently, for any $u$ that is bad for $A_1$ and $B_1$, we have
\[
\begin{array}{ll}
K(u) & \leq t_2 - 6 \log n + 2(\log t_1 + \log n) + O(1) \\
& \leq t_2 - \log n.
\end{array}
\]
Since $K(y) = t_2$, it follows that $y$ is good for $A_1$ and $B_1$.

Let $G$ be the subset of $B_1$ of positions in the strip $B_1 \times \{y\}$ that have $A_2$ colored cells.
Since $A_2 \subseteq A_1$, the $A_2$-colored cells form a subset of the $A_1$-colored cells.

Since $Ext(x,y) = z \in A_2$ and $x \in B_1$, the cell $(x,y)$ belongs to the rectangle $B_1 \times \{y\}$ and is $A_2$ colored. In other words, $x \in G$.

Taking into account that $y$ is good for $A_1$ and $B_1$, we can bound the size of $G$ by
\[
\begin{array}{rl}
n^6 \bigg( \frac{|A_1|}{M} & \cdot 2^{((\delta/2) \log 1/\epsilon)^c} +\epsilon \bigg) |B_1| \\
& \leq 2^{t_1 + 1 + 6 \log n} \bigg( \frac{2^{m-t}}{2^m} \cdot 2^{((\delta/2) \log 1/\epsilon)^c} + 2^{-\log 1/\epsilon} \bigg) \\
& = 2^{t_1 - t + 6 \log n + ((\delta/2) \log n)^c + 1} + 2^{t_1 - \log 1/\epsilon + 1 + 6 \log n} \\
& \leq 2^{t_1 - 6 \log n}.

\end{array}
\]
The set $G$ can be enumerated if we are given $y$, $t_1$, $\delta$ and $n$ (from which we can derive $t$ and the table $Ext$), and every element in $G$ can be described by its rank in the enumeration and by the information neeeded to perform the enumeration.

Since $x \in G$, it follows that
\[
\begin{array}{ll}
K(x \mid y) & \leq t_1 - 6 \log n + 2 (\log t_1 + \log n) + O(1) \\
& < t_1 - \log n,
\end{array}
\]
which contradicts that $K(x \mid y) \geq K(x) - \log n$. Thus $K(z \mid y) \geq m - t = m - O((\log n)^c)$.~\qed
\end{proof}
\fi
\begin{theorem}
For every $\delta > 0$ and for every function $\alpha(n)$, there exists a constant $c$ and a polynomial-time computable function
$f: \zo^n \times \zo^n \mapping \zo^m$, where $m = \Omega(\delta n)$, with the following property. If $n$ is sufficiently large and  $x$ and $y$ are two strings of length $n$ satisfying

(i) $K(x) \geq \delta n$

(ii) $K(y) \geq \delta n$

(iii) $x$ and $y$ are at most $\alpha(n)$ - dependent,

then
\[
K(f(x,y)) \geq m - (\alpha(n) + O((\log n)^c)). 
\]
\end{theorem}

\begin{proof}
Let $\epsilon = 1/(8n^{10}\cdot \alpha(n))$ and let $c$ be the constant and  $Ext : \zo^n \times \zo^n \mapping \zo^m$ be the function given by Theorem~\ref{t:condenser} for parameters $(\delta/2)$ and $\epsilon$. Let $t = \alpha(n) + 10\log n + ((\delta/2) \log 1/\epsilon)^c + 3 = \alpha(n) + O((\log n)^c)$. 

The function $f$ on input $x$ and $y$ returns $z = Ext(x,y)$. 
We show that $K(z )  \geq m-t$.  

Suppose $K(z ) < m-t$.

Let $t_1 = K(x)$, $t_2 = K(y)$, $B_1 = \{u \in \zo^n \mid K(u) \leq t_1\}$, $B_2 = \{u \in \zo^n \mid K(u) \leq t_2\}$.
From hypothesis, $t_1 \geq \delta n$ and $t_2 \geq \delta n$.

Note also that $2^{\delta n/2} \leq |B_1| \leq 2^{t_1 +1}$ and $2^{\delta n/2} \leq |B_2| \leq 2^{t_2 +1}$. (The sets $B_1$ and $B_2$ have size $\geq 2^{\delta n/2}$ because they contain $0^{n-\delta n/2} \{0,1\}^{\delta n/2}$.)

Let
\[
A = \{v \in \zo^m \mid K(v) < m-t\}
\]
We focus on the table defined by the function $Ext: [N] \times [N] \mapping [M]$, where, as usual, we have identified $\zo^n$ with $[N]$ and
$\zo^m$ with $[M]$.

Let $G$ be the subset of $B_1 \times B_2$ of cells in the rectangle $B_1 \times B_2$ that are $A$-colored.

Since $Ext(x,y) = z \in A$, $x \in B_1$ and $y \in B_2$, the cell $(x,y)$ belongs to the rectangle $B_1 \times B_2$ and is $A$-colored. In other words, $x \in G$.

Taking into account Lemma~\ref{l:balanceRao}, we can bound the size of $G$ by
\[
\begin{array}{rl}
\bigg( \frac{|A|}{M} \hspace{-0.1cm} & \cdot 2^{((\delta/2) \log 1/\epsilon)^c} +\epsilon \bigg) |B_1 \times B_2| \\
& \leq 2^{t_1 + t_2 + 2} \bigg( \frac{2^{m-t}}{2^m} \cdot 2^{((\delta/2) \log 1/\epsilon)^c} + 2^{-\log 1/\epsilon} \bigg) \\
& = 2^{t_1 + t_2 - t + ((\delta/2) \log 1/\epsilon)^c + 2} + 2^{t_1 + t_2 - \log 1/\epsilon + 2} \\
& \leq 2^{t_1 + t_2 - (\alpha(n) + 10 \log n)}.

\end{array}
\]
The last inequality follows from the choice of $\epsilon$ and $t$.

The set $G$ can be enumerated if we are given $t_1$, $t_2$, $\delta$ and $n$ (from which we can derive $t$ and the table $Ext$), and every element in $G$ can be described by its rank in the enumeration and by the information neeeded to perform the enumeration.

Since $x \in G$, it follows that
\[
\begin{array}{ll}
K(x y) & \leq t_1 +t_2 - \alpha(n) - 10 \log n + 2 (\log t_1 + \log t_2 + \log n) + O(1) \\
& < t_1 + t_2 - \alpha(n) - 4 \log n \\
& = K(x) + K(y) -\alpha(n) - 4 \log n.
\end{array}
\]
We have used the fact that $t_1 \leq n + O(1)$ and $t_2 \leq n + O(2)$. By the Symmetry of Information Theorem,
\[
K(xy) \geq K(y) + K(x \mid y) - 2\log n - O(1).
\]
Combining the last two inequalities, we get
\[
K(x) - K(x \mid y) > \alpha(n) + \log n,
\]
which contradicts the fact that $x$ and $y$ are at most $\alpha(n)$-dependent.~\qed
\end{proof}

\newcommand{\etalchar}[1]{$^{#1}$}


\end{document}